\documentclass[11pt]{article} 

\usepackage[utf8]{inputenc} 
\usepackage{geometry} 
\usepackage{graphicx} 

\usepackage{array} 
\usepackage{verbatim} 
\usepackage{hyperref}
\usepackage{fancyhdr} 
\usepackage{amsmath}
\usepackage{amsfonts}
\usepackage{amssymb}
\usepackage{amsthm}

\usepackage[lined,algonl,boxed]{algorithm2e}

\newtheorem{theorem}{Theorem}
\newtheorem{lemma}{Lemma}

\newtheorem{definition}{Definition}

\newtheorem{claim}{Claim}

\newcommand{\OPT}{\textit{OPT}}

\title{Maximizing Non-monotone Submodular Set Functions Subject to Different Constraints: Combined Algorithms}
\author{Salman Fadaei\footnote{Faculty of Engineering, University of Tehran, Tehran, Iran. Email: salman.fadaei@gmail.com}
\and MohammadAmin Fazli\footnote{Department of Computer Engineering, Sharif University of Technology, Tehran, Iran. Email: fazli@ce.sharif.edu}
\and MohammadAli Safari\footnote{Department of Computer Engineering, Sharif University of Technology, Tehran, Iran. Email: msafari@sharif.edu} 
 }
\begin{document}
\maketitle

We study the problem of maximizing constrained non-monotone submodular functions and provide approximation algorithms that  improve  existing algorithms in terms of either the approximation factor or simplicity. Our algorithms combine existing local search and greedy based algorithms.
Different constraints that we study are exact cardinality and multiple knapsack constraints. For the multiple-knapsack constraints we achieve a $(0.25-2\epsilon)$-factor algorithm.

We also show, as our main contribution, how to use the continuous greedy process for non-monotone functions and, as a result, obtain a $0.13$-factor approximation algorithm for maximization over any solvable down-monotone polytope.
The continuous greedy process  has  been previously used for maximizing smooth monotone submodular function over a down-monotone polytope \cite{CCPV08}. 
This implies a $0.13$-approximation for several discrete problems, such as maximizing a non-negative submodular function subject to a matroid constraint and/or multiple knapsack constraints.


\section{Introduction}
Submodularity is the discrete analogous of convexity. Submodular set functions naturally arise in several different important problems including cuts in graphs \cite{IFF01, GW94}, rank functions of matroids \cite{E70}, and set covering problems \cite{F98}. The problem of maximizing a submodular function is NP-hard as it generalizes  many important problems such as Maximum Cut \cite{FG95}, Maximum Facility Location \cite{B03, AS99}, and the Quadratic Cost Partition Problem with non-negative edge weights \cite{GG05}.

\begin{definition}
A function $f \colon 2^X \to \mathbb{R_+}$ is called submodular if and only if $ \forall A,B\subseteq X$, $f\left(A \right) + f \left(B \right) \geq f \left(A \cap B \right) + f \left(A \cup B \right)$. An alternative definition is that the marginal values of items should be non-increasing, i.e., $ \forall A,B\subseteq X $, $A \subseteq B \subseteq X$ and $x \in X\setminus B$, $f_A(x) \geq f_B(x)$, where $f_A(x)=f(A\cup \{x\})-f(A)$; $f_A(x)$ is called  the marginal value of $x$ with respect to $A$.  
\end{definition}
The Submodular Maximization Problem is a pair $(f, \Delta)$ where $f$ is a submodular function and $\Delta$ is the search domain. Our aim is to find a set $A^{*} \in \Delta$ whose value,  $f(A^{*})$, is maximum. Our focus is on non-monotone submodular functions, i.e., we do not require that  $f(A) \leq f(B)$ for $A\subseteq B \subseteq X$.
\begin{definition}
A packing polytope is a polytope $P \subseteq [0, 1]^X$ that is down-monotone: If $x, y \in [0, 1]^X$ with $x \preceq y$ and $y \in P$, then $x \in P$. A polytope $P$ is solvable if we can maximize linear functions over $P$ in polynomial time~\cite{Sch03}. 
\end{definition}
A packing polytope constraint binds the search domain ($\Delta)$ to a packing polytope. 
\begin{definition}

For a ground set $X$, $k$ weight vectors $\{w^i\}_{i=1}^{k}$, and k knapsack capacities $\{C_i\}_{i=1}^{k}$ are given. A set $V \subseteq X$ is called packable if $\sum_{j \in V} w^i_j \leq C_i$, for  $i=1,\ldots,k$.
\end{definition}
The multiple knapsacks constraint forces us to bind search domain to packable subsets of $X$. In the exact cardinality constraint, we have $\Delta = \{ S\subseteq X\colon \hspace{0.1cm}|S| = k\}$.

\paragraph{\bf Background}
The problem of  maximizing  non-monotone submodular functions, with or without some constraints,  has been extensively studied in the literature. In \cite{FMV07} a $0.4$-factor approximation algorithm was developed for maximizing unconstrained (non-negative, non-monotone) submodular functions. The approximation factor was very recently improved to $0.41$ by Oveis Gharan et al.\cite{OV11}.

For the constrained variants, Lee et al.~\cite{LMNS09}, Vondr{\'a}k \cite{V09}, and Gupta et al.~\cite{GRST10} provide the best approximation algorithms. Lee et al.~\cite{LMNS09} developed a $0.2$-approximation  for the problem subject to a constant number of knapsack constraints, followed by  a $0.25$-approximation for cardinality constraint and a $0.15$-approximation for the exact cardinality constraint. 
The latter two approximation factors  were later improved by Vondr{\'a}k~\cite{V09} to $0.309$ and $0.25$, respectively. As a new way of tackling these problems, Gupta et al.~\cite{GRST10} provide greedy algorithms that achieve the approximation factor of $0.17$ for a knapsack constraint. Greedy algorithms are more common for maximizing monotone submodular functions.

In a recent work, Vondr{\'a}k~\cite{V08} and Calinescu et al.~\cite{CCPV08} used the idea of multilinear extension of submodular functions and achieved optimal approximation algorithms for the problem of maximizing a monotone submodular function subject to a matroid.

\subsection{Our Results}
We consider the problem subject to different constraints. Our results are summarized in Table \ref{tab:our-results-comparison} and are  compared with existing results.
We obtain simple  algorithms for the exact cardinality constraint, multiple knapsack constraints, and a new approximation algorithm for the solvable packing polytope constraint.
\begin{table}
\centering
\begin{tabular}{|l|c|c|c|c|l|}
\hline
\textbf{Constraint} & \textbf{\cite{LMNS09}} & \textbf{\cite{V09}} & \textbf{\cite{GRST10}} & \textbf{Our Result} & \textbf{Claim}\\
\hline
\hline
Exact Cardinality & $0.15$ & $0.25$ & $0.17$ & $0.25$ & Simpler\\
\hline 
k-Knapsacks & $0.2$ & - & - & $0.25$ & Better ratio\\
\hline 
Packing Polytope & - & - & - & $0.13$ & New ratio\\
\hline 
\end{tabular}
\caption{Comparison of our results with the existing ones. 
}
\label{tab:our-results-comparison}
\end{table}

\subsection{Preliminaries}
In this section, we introduce the concepts and terms that we often use throughout this paper.
\paragraph{Multilinear Extension}
For a submodular function $f\colon 2^X \rightarrow \mathbb{R}_+$, the multilinear extension of $f$ is defined as follows \cite{CCPV07}: $F\colon [0,1]^X \rightarrow \mathbb{R}_+$ and
\begin{displaymath}
F(x)=\textbf{E}[f(x)]= \sum_{S \subseteq X}f(S) \prod_{i \in S}x_i \prod_{i \in X\setminus S}(1-x_i).
\end{displaymath}
This concept is frequently used in recent works \cite{CCPV07, CCPV08, KST09, LMNS09, V09}. The multilinear extension of every submodular function is a smooth submodular function \cite{CCPV08}.
The gradient of $F$ is defined as  $\nabla F=(\frac {\partial F}{\partial x_1},\ldots, \frac{\partial F}{\partial x_n})$.
\paragraph{Matroid} A matroid is a pair $\mathcal{M} = (X, \mathcal{I})$ where $\mathcal{I} \subseteq 2^{X}$ and 
\begin{itemize}
\item $\forall B\in \mathcal{I}, A \subset B \Rightarrow A \in \mathcal{I} $.
\item $\forall A,B \in \mathcal{I}, |A| < |B| \Rightarrow \exists x \in B\backslash A; A \cup \{x\} \in \mathcal{I}$.
\end{itemize}
\paragraph{Matroid Polytopes}
A matroid polytope is a solvable packing polytope with special properties. Given a matroid $\mathcal{M}=(X,\mathcal{I})$, we define the matroid polytope as 
\begin{displaymath}
P(\mathcal{M})=\{x \geq 0 \colon \forall S \subseteq X;\sum_{j \in S} x_j \leq r_{\mathcal{M}}(S)\}
\end{displaymath}
where $r_{\mathcal{M}}(S)=max \{|I| \colon I \subseteq S; I \in \mathcal{I}\}$ is the rank function of matroid $\mathcal{M}$. This definition shows that the matroid polytope is a packing polytope.

\paragraph{Randomized Pipage Rounding}
Given a matroid $\mathcal{M}=(X,\mathcal{I})$, the randomized pipage rounding converts a fractional point in the matroid  polytope, $y \in P(\mathcal{M})$ into a random set $B\in \mathcal{I}$ such that $\textbf{E}[f(B)] \geq F(y)$, where $F$ is the multilinear extension of the submodular function $f$~\cite{CCPV07, CCPV08, V09}.

\subsection{Recent Developments}
There has been some very recent relevant works independent and concurrent to our work.
Kulik et al. give an $(0.25-\epsilon)$-approximation algorithm for maximizing non-monotone submodular functions subject to multiple knapsacks~\cite{KST11}. Chekuri et al.~\cite{CVZ11} show that, by using a fractional local search, a $0.325$-approximation could be achieved for maximizing non-monotone submodular functions subject to any solvable packing polytope. However, our  $0.13$-factor approximation algorithm is still of independent interest in that it uses the continuous greedy approach rather than local search and, thus, it would be more efficient in practice.

\section{Exact Cardinality Constraint}
In this section, we propose very simple algorithm for the exact cardinality constraint problem whose approximation factor matches the best existing one, yet it is much simpler and easy to implement. 
Our algorithm is a simple combination of existing local search or greedy based algorithms.
Our main tool is the following useful lemma from Gupta et al.~\cite{GRST10}.
\begin{lemma} \label{lem-three-subsets}
(\cite{GRST10}) Given sets $C, S_1 \subseteq X$, let $C' = C\setminus S_1$ and $S_2 \subseteq  X\setminus S_1$.
Then, $f(S_1 \cup C) + f(S_1 \cap C) + f(S_2 \cup C') \geq f(C)$.
\end{lemma}

Let $k$  be the right-hand side of the cardinality constraint.
\begin{theorem} \label{thm-exact-card-non-mono}
There is a $0.25$-factor approximation algorithm for maximizing a non-monotone submodular function subject to an exact cardinality constraint.
\end{theorem}
\begin{proof}
First, we use the local search algorithm of  \cite{LMNS09} and compute  a set
$S_1$ whose size is $k$ and 
$2f(S_1) \geq f(S_1 \cup C) + f(S_1 \cap C)$ for any $C$ with $|C| = |S_1| = k$.
Next, we use the greedy algorithm of \cite{GRST10} and compute a set $S_2\subseteq X\backslash S_1$ of size $k$ such that 
for any  $C'$ with $|C'| \leq k$, $f(S_2) \geq 0.5f(S_2\cup C')$. Let $C$ be the true optimum and $C' = C\backslash S_1$. Therefore, 
 $$
 2f(S_1) + 2f(S_2) \geq f(S_1 \cup C) + f(S_1 \cap C) + f(S_2\cup C') \geq f(C) = OPT
 $$
Thus, the better of $S_1$ and $S_2$ gives an approximation factor $0.25$.

Here, we have assumed that $k \leq \frac{|X|}{2}$. If not, we can alternatively solve the problem for the derived submodular function $g(S) = f(X \setminus S)$ subject to cardinality constraint $k' = |X| - k$.
\end{proof}
The approximation factor $0.25$ matches that of \cite{V09}, though our algorithm is simpler and  straightforward to implement.

\section{Multiple Knapsack Constraints} \label{k-knapsack}
Lee et al.~\cite{LMNS09} propose a $0.2$-factor approximation algorithm for the problem.
They basically divide the elements into two sets of heavy and light objects and then solve the problem separately for each set  and return the maximum of the two solutions.

 We improve their result by considering both heavy and light elements together. Our algorithm finds a fractional solution and then integrates it by using independent rounding. We use some of the properties of the independent rounding; For the sake of completeness, we mention it  before presenting the main algorithm. 
 
 Let $x = (x_1,\ldots, x_n)$ be a fractional solution and $(X_1,\ldots, X_n) \in \{0,1\}^n$ be an integral solution obtained from $x$ by  randomized independent rounding. We observe that $\textbf{E}[X_i]=x_i$ and for any subset $T$, $\textbf{E}[\prod_{i \in T}X_i]=\prod_{i \in T} x_i$, and $\textbf{E}[\prod_{i \in T}(1-X_i)]=\prod_{i \in T} (1-x_i)$. Considering these properties, as in \cite{CVZ10} (Theorem II.1) and \cite{GKPS06} (Theorem 3.1), we obtain the following Chernoff-type concentration bound for  linear functions of $X_1,\ldots,X_n$.

\begin{lemma} \label{lem-cor-1-2-cvz09}
Let $a_i \in [0,1]$ and $X=\sum a_i X_i$ where $(X_1,\ldots,X_n)$ are obtained by randomized independent rounding from a starting point $(x_1,\ldots,x_n)$. 
Then (i) for $\delta \in [0,1]$, and $\mu \geq \textbf{E}[X]=\sum a_i x_i$, we have $Pr[X\geq (1+\delta) \mu] \leq e ^{-\mu \delta ^2 /3}$ and (ii) for $\delta \geq 1$, $Pr[X\geq (1+\delta) \mu] \leq e ^{-\mu \delta /3}$.
\end{lemma} 

The following simple observation will be useful in the presentation of the algorithm.

\begin{lemma} \label{lem-union-subm-function}
Let $f \colon 2^X \to \mathbb{R_+}$ be a nonnegative submodular function over $X$, $A \subseteq X$ and $X'=X\setminus A$, then $h \colon 2^{X'} \to \mathbb{R_+}$ where $h(S)=f(A\cup S)$, $\forall S \subseteq X'$ is a nonnegative submodular function.
\end{lemma}

\begin{proof}
Fix an arbitrary $A \subseteq X$ and define $X'=X\setminus A$. Let $S$ and $T$ with $S\subseteq T$ be two arbitrary subsets of $X'$. Let $x \in X'\setminus T$.
To show submodularity of $h$, we need show that $h_S(x) \geq h_T(x)$.
Recalling the definition of $h$, it is equivalent to show $f(S\cup A \cup \{x\})-f(S\cup A) \geq f(T\cup A \cup \{x\})-f(T\cup A)$. 
But, this is true since $S\cup A \subseteq T\cup A$ and $f$ is submodular.
The nonnegativity of $h$ is clear, the desired conclusion.
\end{proof}

Our algorithm (Algorithm \ref{alg1} below) is  based on the algorithm of Chekuri et al.~\cite{CVZ10} for maximizing  monotone submodular functions subject to one matroid and multiple knapsack constraints. We have made some modifications to use it for non-monotone functions. 

\begin{algorithm} \label{alg1}
\KwIn{Elements weights $\{c_{i,j}\}$, parameter $0 <\epsilon < 1/(4k^2)$, and a non-monotone submodular function $f$}

 $D \leftarrow \emptyset$.

\ForEach{subset $A$ of at most $1/\epsilon ^4$ elements}{
0. Set $D\leftarrow A$ if $f(A) > f(D)$\;

1. Redefine $C_j=1-\sum_{i \in A}{c_{ij}}$ for $1 \leq j \leq k$.\;

2. Let $B$ be the set of items $i \notin A$ such that either $f_A(i) > \epsilon ^4f(A)$ or $ c_{ij} > k\epsilon ^3Cj$ for some $j$\;

3. Let $x^*$ be the fractional solution of the following problem:
  \begin{equation} \label{prob-free-mat}
  \max\{H(x) :  x \in [0,1]^{X'}; \hspace{5pt}  \forall j \sum {c_{ij}x_i \leq (1-\epsilon)C_j}\}
  \end{equation}
  where $X' = X\setminus (A\cup B)$, and $H(x)$ is the multilinear extension of $h(S)=f(A\cup S)$, $\forall S \subseteq X'$\;
  
 4. Let $R$ be the result of the independent rounding applied to $x^*$;\\ 
If $\exists j:\sum_i c_{ij}x_i > C_j$  then $S \leftarrow \emptyset$ else $S \leftarrow R$; 
 Set $D\leftarrow A\cup S$ if $f(A\cup S) > f(D)$\;
}
 Return $D$.
 \caption{Non-Monotone Maximization Subject to Multiple Knapsacks}
\end{algorithm}

The following theorem shows how good our algorithm is.
\begin {theorem} \label{thm-kknap-non-mono}
Algorithm 1 returns a solution of expected value at least $(0.25-2\epsilon)\OPT$.
\end{theorem} 
\begin{proof}
The proof follows the line of proofs of \cite{CVZ10} with major changes to adapt it for non-monotone case.
Let $O$ be the optimal solution and $\OPT = f(O)$.
Assume $|O|\geq \frac{1}{\epsilon^4}$; otherwise, our algorithm finds the optimal solution in Line 0.
Sort the elements of $O$ by their decreasing marginal values, and let $A\subseteq O$ be 
the first $ \frac{1}{\epsilon^4}$ elements. Consider the iteration in which  this set $A$ is chosen.
Since $A$ has $\frac{1}{\epsilon^4}$ elements, the marginal value of its last element and every element not in $A$ is at most $\epsilon^4f(A) \leq \epsilon^4 \OPT$. 
Hence, throwing away elements whose marginal value is bigger than $\epsilon^4f(A)$ does not affect items in $O\setminus A$.
We also throw away the set $B \subseteq X\setminus A$ of items whose size in some knapsack is more than $k\epsilon^3Cj$. 
In $O\setminus A$, there can be at most $1/(k\epsilon^3)$ such items for each knapsack, i.e., $1/\epsilon^3$ items in total. 
Since their marginal values with respect to $A$ and consequently w.r.t. $A\cup O'=O\setminus B$ are bounded by $\epsilon^4\OPT$ (where $O'=O\setminus (A\cup B)$), by submodularity these items together have marginal value $f(O)-f(O\setminus B) \leq \epsilon OPT$, therefore $f(O\setminus B) \geq (1-\epsilon) OPT$. 
For set $O'$ we have:
\begin{displaymath}
h(O') = f(A\cup O')=f(O\setminus B) \geq (1-\epsilon) OPT.
\end{displaymath}
The indicator vector $(1-\epsilon)1_{O'}$ is a feasible solution for Problem \ref{prob-free-mat} (specified at  step 3 of Algorithm \ref{alg1}). By the concavity of $H(x)$ along the line from the origin to $1_{O'}$, we have \\ $H\big((1-\epsilon)1_{O'}\big) \geq (1-\epsilon)h(O')\geq (1-2\epsilon)OPT$. By Theorem 4 of \cite{LMNS09} we can compute in polynomial time a fractional solution $x^*$ with value:
\begin{displaymath}
H(x^*) \geq \frac{1}{4}H\big((1-\epsilon)1_{O'}\big) \geq (\frac{1}{4}-\epsilon)OPT.
\end{displaymath}
Notice, according to Lemma \ref{lem-union-subm-function}, $h$ is a nonnegative submodular function and we can apply the algorithm of ~\cite{LMNS09}.
Finally, we apply independent rounding to $x^*$ and call the resulting set $R$. By the construction of independent rounding, we have $\textbf{E}[h(R)] = H(x^*)$. However, $R$ might
violate some of the knapsack constraints.

Define $P(l)=\{x \in [0,1]^{X'}; \hspace{5pt}  \forall j \sum {c_{ij}x_i \leq lC_j}\}$ for $l\geq 1$ and $l$ is integer.
Define $A_1$ as the event that $\mathbf{1}_R \in P(1)$.
By definition of $S$, we have $\textbf{E}[h(S)]=\textbf{E}[h(R)|A_1]Pr[A_1]$. 
Analogously, define disjoint events $A_l$ such that $\mathbf{1}_R \in P(l)\setminus P(l-1)$ for $l \geq 2$ and $l$ is integer.

We have

\begin{equation} \label{outcome-of-events-knap}
H(x^*)=\textbf{E}[h(R)]=\sum_{l \geq 1} \textbf{E}[h(R)|A_l]Pr[A_l].
\end{equation}

Consider a fixed knapsack constraint $j$.
Our fractional solution $x^*$ satisfies $\sum{c_{ij}x^*_i} \leq (1-\epsilon)C_j$. 
Also, we know that all sizes in the reduced instance are bounded by $c_{ij}\leq k \epsilon ^3 C_j$. 
By scaling, $c'_{ij}=c_{ij}/(k\epsilon ^3 C_j)$, 
we use Lemma \ref{lem-cor-1-2-cvz09} ($i$)  with 
$\mu=(1-\epsilon)/(k\epsilon ^3)$:
\begin{displaymath}
Pr[\sum_{i \in R}c_{ij}>C_j] \leq Pr[\sum_{i \in R}c'_{ij}>(1+\epsilon)\mu] \leq e^{-\mu \epsilon ^2/3}<e^{-1/4k \epsilon}.
\end{displaymath}

From this, $Pr[A_1]\geq 1-ke^{-1/4k\epsilon}$ and $Pr[A_2]\leq ke^{-1/4k\epsilon}$.

Similarly, we calculate the probability of events $A_l$, $l\geq 3$.
Let $\delta=(l-2+\epsilon)/(1-\epsilon)$. 
From Lemma \ref{lem-cor-1-2-cvz09} $(ii)$ and using $\mu=(1-\epsilon)/(k\epsilon^3)$,

\begin{displaymath}
Pr[\sum_{i \in R}c_{ij}>(l-1)C_j] \leq Pr[\sum_{i \in R}c'_{ij}>(1+\delta)\mu] \leq e^{-\mu \delta /3}\leq e^{-(l-2)/ (k\epsilon^3)}\leq e^{-l/(3k \epsilon^3)}.
\end{displaymath}

Using the union bound, we can write for any $l\geq 3$

\begin{equation*}
Pr[A_l]\leq Pr[\exists j; \sum_{i\in R}c_{ij}>(l-1)C_j]\leq ke^{-l/(3k \epsilon^3)}.
\end{equation*} 

Since $H$ is concave along rays through the origin, for all $l\geq 3$ we obtain

\begin{equation*}
\begin{array}{ll}

\textbf{E}[h(R)|A_l] & \leq \max \{H(x)|\ x \in P(l)\} \\ 
& \leq  4lH(x^*).

\end{array}
\end{equation*}

Plugging our bounds into \ref{outcome-of-events-knap} we obtain
\begin{displaymath}
\begin{array}{ll}
H(x^*)
&=\textbf{E}[h(R)|A_1]Pr[A_1]+\sum_{l\geq 2} \textbf{E}[h(R)|A_l]Pr[A_l]\\
&\leq \textbf{E}[h(R)|A_1]Pr[A_1]+8kH(x^*)e^{-1/(4k\epsilon)}+4kH(x^*)\sum_{l\geq 3}le^{-l/(3k\epsilon^3)}\\
&=\textbf{E}[h(S)]+4kH(x^*)(2e^{-1/(4k\epsilon)}+\sum_{l\geq 3}le^{-l/(3k\epsilon^3)}).
\end{array}
\end{displaymath}

Let $q=e^{-1/(3k\epsilon^3)}$. 
Using formula $\sum_{l\geq 3} lq^l=\frac{3q^3}{1-q}+\frac{q^4}{(1-q)^2} \leq q^{3\epsilon^2}$, we obtain $\sum_{l\geq 3}le^{-l/(3k\epsilon^3)}\leq e^{-1/(k\epsilon)}$.

Therefore, we obtain
$H(x^*) \leq \textbf{E}[h(S)]+ 4kH(x^*)e^{-1/k\epsilon}$.
Using $4ke^{-1/k\epsilon}<\epsilon$ for $\epsilon<\min\{1/k^2,0.001\}$, we get $\textbf{E}[h(S)]\geq (1-\epsilon)H(x^*)$.
Therefore, we have a feasible solution of expected value $\textbf{E}[f(S\cup A)] \geq (1-\epsilon)H(x^*) \geq (1/4-2 \epsilon) OPT$.

\end{proof}

\section{Packing Polytope Constraint} \label{greedy-process}
In this section, we adapt the  continuous greedy process  for non-monotone submodular functions and propose an algorithm for solving the optimization problems subject to a packing polytope constraint. As an application of the technique, we then consider the problem of submodular maximization subject to both one matroid and multiple knapsacks constraints. Finally, we  briefly show how to replace this continuous process with a polynomial time discrete process without suffering much.

\subsection{Continuous greedy process for non-monotone functions} \label{contiprocess}
Similar to  \cite{CCPV08}, the greedy process starts with $y(0) = \textbf{0}$ and  increases over a unit time interval as follows:
\begin{displaymath}
\frac {dy}{dt} = v_{max}(y),
\end{displaymath}
where $v_{max}(y)=argmax_{v \in P}(v.\nabla F(y)).$
When $F$ is a non-monotone smooth submodular function, we have
\begin{lemma} \label{lem-greedy-proc-non-mono}
$y(1) \in P$ and $F(y(1)) \geq (1-e^{-1}) (F(x \vee y(1))-F_{DMAX})$, where $x \in P$, and $F_{DMAX} = \max_{0 \leq t \leq 1}{F(y(1)-y(t))}$.
\end {lemma}
\begin{proof} The proof is essentially similar to that of \cite{CCPV08} with some modifications to adapt it for non-monotone functions. 
 First, the trajectory for $t \in [0, 1]$ is contained in P since
\begin{displaymath}
y(t)=\int _{0}^{t}{v_{max}(y(\tau))d\tau}
\end{displaymath}
is a convex linear combination of vectors in $P$. To prove the approximation guarantee, fix a point
$y$.
Consider a direction $v^* = (x \vee y)-y = (x-y) \vee 0$. This is a non-negative vector; since $v^* \leq x \in P$ and $P$ is down-monotone, we also have $v^* \in P$. Consider the ray of direction $v^*$ starting at $y$, and the function $F(y + \xi v^*)$, $\xi \geq 0$. The directional derivative of $F$ along this ray is $\frac {dF}{d \xi}=v^*.\nabla F$. Since $F$ is smooth submodular (that means, each entry $\frac{\partial F}{\partial y_j}$ of $\nabla F$ is non-increasing with respect to $y_j$) and $v^*$ is nonnegative,  $\frac {dF}{d \xi}$ is non-increasing too and $F(y + \xi v^*)$ is concave in $\xi$. By concavity, we have
\begin {displaymath}
F(y(1) + v^*)-F(y(t)) \leq F(y(t) + v^*)-F(y(t)) + F(y(1)-y(t)) \leq v^* .\nabla F(y(t)) + F_{DMAX}. 
\end {displaymath}
Since $v^* \in P$ and $v_{max}(y) \in P$ maximizes
$v.\nabla F(y)$ over all vectors $v \in P$, we get
\begin{equation} \label{diff-inequ}
v_{max}(y).\nabla F(y) \geq v^*.\nabla F(y) \geq F(y(1) + v^*)-F_{DMAX}- F(y).
\end{equation}

We now get back to  the continuous process and analyze $F(y(t))$. Using the chain rule and the inequality (\ref{diff-inequ}), we get
\begin{displaymath} 
\frac {dF}{dt}=\sum _{j}{\frac{\partial F}{\partial y_j}\frac{dy_j}{dt}}=v_{max}(y(t)).\nabla F(y(t)) \geq F(x \vee y(1))-F_{DMAX}- F(y(t)).
\end{displaymath}
Thus,  $F(y(t))$ dominates the solution of the differential equation 
$$\frac{d \phi}{dt}=F(x \vee y(1))-F_{DMAX} - \phi(t)$$
 which means	 $\phi(t) = (1-e^{-t})(F(x \vee y(1))-F_{DMAX})$. Therefore, $F(y(t)) \geq (1-e^{-t})(F(x \vee y(1))-F_{DMAX})$.
\end{proof}

\subsection{Extending Smooth Local Search.}
As our final tool for obtaining the main algorithm of this section, we propose an algorithm for the following problem:
Let $f$ be a sumbodular function and $F$ be its multilinear extension. Let $u_i\in [0, 1]$, $1\leq i\leq n$, be a set of upper bound variables and $\mathcal{U}:= \{0\leq y_i \leq u_i \hspace{0.25cm} \forall i \in X\}$.
We want to maximize $F$ over the region $\mathcal{U}$:

\begin{displaymath}
\max \{F(y) \colon y \in \mathcal{U}\}
\end{displaymath} 

For this, we extend the $0.4$-approximation algorithm (Smooth Local Search or SLS) of \cite{FMV07} as follows. We call our algorithm $FMV_Y$.

We define a discrete set $\zeta$  of values in $[0,1]$, where 	$\zeta=\{p.\delta \colon 0 \leq p \leq 1/\delta \}$, $\delta=\frac{1}{8n^4}$ and $p$ is integer. The algorithm returns a vector whose values come from the discrete set $\zeta$. We show that such a discretization  does not substantially harm our solution, yet it reduces the running time.

Let $U$ be a multiset containing $s_i=\lfloor\frac{1}{\delta} u_i\rfloor$ copies of each element $i\in X$.
We define a set function $g\colon 2^U \rightarrow \mathbb{R}_+$ with $g(T) = F(\ldots, \frac{|T_i|}{s_i}, \ldots)$, where $T\subseteq U$ and $T_i$ contains all copies of $i$ in $T$. The function $g$ has been previously introduced in \cite{LMNS09} and proved to be submodular.
Let $B$ be the solution of running the SLS algorithms for maximizing $g$ and $y$ be its corresponding vector.


Based on \cite{FMV07}, we have  $g(B) \geq 0.4 g(A), \hspace{0.25cm} \forall A \in U$; thus
\begin{equation} \label{fmv_equ}
F(y) \geq 0.4 F(z), \hspace{0.25cm} \forall z \in \mathcal{U}\cap \zeta ^n.
\end{equation}
and we can prove the following claim.
\begin{claim} \label{claim-fmv-y}
For any $x \in \mathcal{U}$, $2.5 F(y) \geq F(x)-\frac{f_{max}}{4n^2}$, where $f_{max}=\max \{f(i) \colon i \in X\}$.
\end{claim}
\begin{proof}
Let $z$ be the point in $\zeta ^n \cap \mathcal{U}$ that minimizes $\sum_{i=1}^n(x_i-z_i) $. By Claim 3 of \cite{LMNS09},  $F(z) \geq F(x) - \frac{f_{max}}{4n^2}$. Using the  inequality (\ref{fmv_equ}), we get $F(y) \geq 0.4(F(x)-\frac{f_{max}}{4n^2})$. This completes the proof.
\end{proof}

\subsection{The Algorithm}

We now present our  algorithm for  maximizing a smooth submodular function over a solvable packing polytope:

\begin{algorithm}[H] \label{alg2}
\KwIn{A packing polytope $P$ and a smooth submodular function $F$}

1. $y_1 \longleftarrow$ The result of running the continuous greedy process.

2. $y'_1 \longleftarrow argmax_{0 \leq t \leq 1} F(y_1-y(t))$. 

3. $y_{1max} \longleftarrow $  The result of running $FMV_Y$ with the upper bound $y_1$.
  
4. $y_2 \longleftarrow $ The result of  running the greedy process  over the new polytope $P'$ which is built by adding constraints $y_i \leq 1-y_{1_{i}}$ for any $1 \leq i \leq n$ to $P$. Note that $P'$ is a down-monotone polytope.
 
 5. $y'_2 \longleftarrow argmax_{0 \leq t \leq 1} F(y_2-y(t))$.
 
 6. Return $argmax(F(y_1), F(y_2), F(y_{1max}), F(y'_1), F(y'_2))$.

 \caption{Continuous greedy process for non-monotone functions}
\end{algorithm}

\begin{theorem}
The above algorithm  is a $\frac{2e-2}{13e-9}$-approximation algorithm for the problem of maximizing a smooth submodular function $F$ over a solvable packing polytope $P$.
\end {theorem}
\begin{proof}
Suppose $x^* \in P$ is the optimum and $F(x^*)=OPT$. By Lemma \ref{lem-greedy-proc-non-mono}, $F(y_1) \geq (1-e^{-1})(F(x^* \vee y_1)-F(y'_1))$. We also have $F(y_2) \geq (1-e^{-1})(F(x' \vee y_2)-F(y'_2))$, where $x'=x^*-(x^* \wedge y_1)$. Note that $x' \in P'$. By Claim \ref{claim-fmv-y}, we also have $F(y_{1max}) \geq 0.4 (F(x^* \wedge y_1) - \frac{f_{max}}{4n^2})$ as $x^* \wedge y_1 \preceq y_1$.\\
By adding up the above inequalities, we get
\begin {displaymath}
\begin{array}{ll}
\frac{e}{e-1}(F(y_1)+F(y_2))+F(y'_1)+F(y'_2)+2.5 F(y_{1max}) &\\
\geq F(x^* \vee y_1) + F(x' \vee y_2)  + F(x^* \wedge y_1) - \frac{f_{max}}{4n^2}& \\ 
\geq F(x^*) - \frac{f_{max}}{4n^2}=OPT - \frac{f_{max}}{4n^2}. &
\end {array}
\end {displaymath}
Therefore, the approximation factor of the algorithm is at least
$\frac{2e-2}{13e-9}OPT$.
\end{proof}

\paragraph{Both one matroid and multiple knapsacks}
As a direct result of the above theorem, we propose the first approximation algorithm 
for  maximizing a submodular function subject to both one matroid and multiple knapsacks.
This problem was solved (approximately) in  \cite{CV09} for monotone submodular functions.

\begin {theorem}
There exists an algorithm with expected value of at least $(\frac{2e-2}{13e-9}-3\epsilon)OPT$ for the problem of maximizing any non-monotone submodular function subject to one matroid and multiple knapsacks.
\end {theorem}
\begin{proof} 
The intersection of the polytopes corresponding  to one matroid and multiple knapsacks is still a solvable packing polytope. Thus, we can achieve a fractional solution by using Algorithm \ref{alg2} together with the enumeration phase (as in Algorithm \ref{alg1}), and then we can round the fractional solution into the integral one using randomized pipage rounding.


Our algorithm is very similar to that of \cite{CV09} with some modifications to adapt it for non-monotone functions. As the two algorithms are similar, we only highlight the modifications to our algorithm.

The algorithm in \cite{CV09} is for maximizing monotone submodular functions subject to one matroid and multiple knapsacks and uses partial enumeration. At each iteration, after getting rid of all items of large value or size, it defines an optimization problem with scaled down constraints. Since the objective function is monotone, the reduced problem at each iteration is solved using continuous greedy algorithm to find a fractional solution within a factor  $1-1/e$ of the optimal.

For our case, we cannot use the continuous greedy algorithm as our function is not monotone. Instead, we use Algorithm \ref{alg2}
 to solve the reduced problem and achieve a fractional solution with approximation factor $\frac{2e-2}{13e-9}$. The final step of the two algorithms are identical. 
 At each iteration, we apply randomized pipage rounding to the fractional solution with respect to contracted matroid of the that iteration. 
The result is the set with maximum objective functions over all iterations. 

Our analysis is similar to that of \cite{CV09} except that our approximation factor for the reduced problem (at each iteration) is $\frac{2e-2}{13e-9}$ as opposed to $1-1/e$ of \cite{CV09}. So, the same analysis works with the two approximation factors exchanged.

Note that, because of considerations in the design of the algorithm, randomized pipage rounding does not violate, with high probability,  the capacity constraints on knapsacks and, therefore, our solution is a feasible one with constant probability. We remark that the argument for the concentration bound in \cite{CV09} is applicable to our analysis, as well.
\end{proof}

\subsection{Discretizing Continuous Process }

In order to obtain a polynomial time, we  discretize the continuous  greedy process for non-monotone functions and show  that by taking small enough time steps, this process only introduces a small error that is negligible and the solution to the differential inequality does not significantly change.


Let $\delta = \frac{1}{n^2}$, and $\zeta = \{p.\delta \colon 0 \leq p \leq 1/\delta\}$ be a set of discrete values.
We set the unit time interval equal to $\delta$ in Algorithm \ref{alg2}, and change lines 2 and 5 of it as follows.
\begin{eqnarray*}
2 &  y'_1 \longleftarrow argmax F(y_1-y(t)), \forall t \in \zeta \\
5 & y'_2 \longleftarrow argmax F(y_2-y(t)),  \forall t \in \zeta \\
\end{eqnarray*}

and obtain the following lemma which is weaker (but not very different) than Lemma \ref{lem-greedy-proc-non-mono}.
\begin {lemma} \label{thm-three-subsets}
$y(1), y'_1 \in P$ and $F(y(1)) \geq (1-e^{-1}) (F(x \vee y(1))-F(y'_1))-o(1)OPT$, where $x \in P$, where $P$ is any solvable packing polytope.
\end {lemma}

\hspace{-0.7cm}{\bf Acknowledgement}

The Authors are grateful to Vahab Mirrokni for his help in all the steps of the preparation of this paper. The first author is also thankful to  Ali Moeini (his M.Sc. advisor), Dara Moazzami, and Jasem Fadaei for their help and advice.
The authors also would like to acknowledge the anonymous referees for the their useful comments and suggestions.

\appendix

\end{document}